\documentclass[twocolumn, superscriptaddress, preprintnumbers, amsmath, amssymb, aps]{revtex4}

\usepackage{graphicx}
\usepackage[caption=false]{subfig}
\usepackage{amsthm}
\usepackage{amsmath}
\usepackage{tensor}
\usepackage{color}
\usepackage{bm}
\usepackage[all]{xy}
\usepackage{tikz}
\usepackage{dsfont}
\usepackage{times,txfonts}
\usepackage{algorithm}
\usepackage{algpseudocode}
\usetikzlibrary{positioning}
\usepackage{braket}
\usepackage{enumitem}
\usepackage{booktabs}
\usepackage{makecell}
\usepackage{multirow}
\usepackage{diagbox}

\usepackage{array}

\newtheorem{Thm}{Theorem}
\newtheorem{Lem}[Thm]{Lemma}

\newtheorem{Cor}[Thm]{Corollary}
\theoremstyle{definition}

\begin{document}

\title{
Quantum-Assisted Unambiguous Discrimination of Pure States with Robustness Bounds
}

\author{Jinhyeok Heo}
\affiliation{
Department of Mathematics and Research Institute for Basic Sciences, Kyung Hee University, Seoul, 02447, Korea
}

\author{Dongkeun Lee}
\affiliation{
Center for Quantum Information R\&D, Korea Institute of Science and Technology Information~(KISTI), Daejeon 34141, Republic of Korea
}
\affiliation{
Division of Quantum Information, KISTI School, Korea University of Science and Technology, Daejeon 34141, Republic of Korea
}

\author{Junghee Ryu}
\affiliation{
Center for Quantum Information R\&D, Korea Institute of Science and Technology Information~(KISTI), Daejeon 34141, Republic of Korea
}
\affiliation{
Division of Quantum Information, KISTI School, Korea University of Science and Technology, Daejeon 34141, Republic of Korea
}

\author{Minjin Choi}
\email{mathcmj89@gmail.com}
\affiliation{
Center for Quantum Information R\&D, Korea Institute of Science and Technology Information~(KISTI), Daejeon 34141, Republic of Korea
}

\date{\today}

\begin{abstract}
We propose a hybrid quantum-classical framework for the unambiguous discrimination of pure states by combining the noisy intermediate-scale quantum semidefinite program solver~(NSS) with the reciprocal-state representation.
In NSS, a given semidefinite program~(SDP) is represented within a subspace spanned by chosen ansatz states, with the required quantities estimated using quantum circuits, and the resulting optimization is performed classically.
This makes NSS naturally applicable to discrimination problems involving circuit-prepared states, as finding a discrimination measurement can be formulated as an SDP.
By choosing the discriminated states themselves as the ansatz and exploiting the reciprocal-state representation, the general matrix-variable NSS formulation is reduced to an optimization over $N$ scalar variables, where $N$ is the number of discriminated states.
The resulting formulation reduces the number of optimization variables and positive-semidefinite constraints compared with the original NSS formulation.
We further establish robustness guarantees against Gram matrix estimation errors, including a rescaling procedure that guarantees a valid positive operator-valued measure~(POVM) under bounded perturbations.
The resulting bounds quantify both the deviation in average success probability and violations of the no-error condition using only the estimated Gram matrix, its SDP solution, and a characterized perturbation bound, without access to the ideal Gram matrix.
Numerical simulations validate the derived robustness bounds and characterize finite-shot scaling. 
We also perform experiments on IBM quantum hardware to estimate pairwise overlaps, and then evaluate the performance of the resulting POVMs.
\end{abstract}


\maketitle

\section{Introduction}
\label{sec:introduction}

Quantum state discrimination is one of the fundamental tasks in quantum information theory.
The construction of a discrimination measurement depends on the criterion adopted for distinguishing nonorthogonal quantum states, which cannot be perfectly discriminated.
Several discrimination strategies have been developed for this purpose, including minimum-error discrimination and unambiguous state discrimination~(USD)~\cite{Barnett2009, Bergou2010, Bae2015}.
In particular, USD eliminates erroneous conclusive outcomes by allowing an additional inconclusive outcome, making it particularly attractive in applications where reliability is more important than always producing a decision.
The strategy was independently introduced by Ivanovic, Dieks, and Peres for two nonorthogonal pure states with equal prior probabilities \cite{Ivanovic1987, Dieks1988, Peres1988}, and was later extended to unequal priors \cite{Jaeger1995}. 
Subsequent works have extended USD to more general settings and multi-state ensembles~\cite{CheflesBarnett1998, Chefles1998, Sugimoto2010, Bergou2012}.
Recent studies have continued to explore USD in different contexts, including approximate unambiguous discrimination~\cite{Zhuang2020}, unambiguous discrimination of quantum operations~\cite{Cai2024}, and sequential unambiguous discrimination~\cite{Perez2025}.

One class of hybrid quantum-classical approaches uses a quantum processor to obtain problem-relevant quantities, which are subsequently processed classically~\cite{Bharti2022, Bharti2022nisq, Haug2022}. 
Within this approach, Bharti \textit{et al.}~\cite{Bharti2022} proposed the noisy intermediate-scale quantum~(NISQ) solver for semidefinite programs~(SDPs), referred to as the NISQ SDP solver~(NSS).
NSS represents an SDP within a subspace spanned by a chosen set of ansatz states.
The resulting SDP is specified by matrices constructed with respect to the ansatz states, whose entries can be estimated using quantum circuits, while the SDP is solved classically.
A key feature of the framework is the choice of ansatz, which determines how the optimization problem is represented and directly influences both the solution quality and the computational complexity.

The NSS framework is naturally applicable to quantum state discrimination, as finding a discrimination measurement can be formulated as an SDP.
In particular, when the states to be discriminated are prepared by quantum circuits, quantities required for the optimization can be estimated using their preparation circuits.
Indeed, the application of NSS to USD was previously illustrated by Bharti \textit{et al.}~\cite{Bharti2022} using the general NSS formulation.
More recently, Mohan \textit{et al.} proposed a related hybrid quantum-classical approach for Bayes-optimal quantum state discrimination~\cite{Mohan2026}.
While these works demonstrate the applicability of hybrid quantum-classical approaches to quantum state discrimination, their general formulations leave room for further simplification by exploiting the known structure of pure-state USD.
The effect of errors in the estimated quantities on the resulting discrimination measurement has also not been characterized.

In this work, we propose a hybrid quantum-classical framework for the unambiguous discrimination of pure states by combining NSS with the reciprocal-state representation of USD~\cite{Eldar2003}.
By choosing the discriminated states themselves as the ansatz and exploiting the reciprocal-state representation, the general matrix-variable NSS formulation reduces to an optimization over $N$ scalar variables, where $N$ is the number of discriminated states.
This reduction simplifies the classical postprocessing compared with the original NSS formulation.
We further analyze the effect of Gram matrix estimation errors on the resulting discrimination measurement.
For a bounded perturbation of the estimated Gram matrix, we introduce a rescaling procedure that guarantees a valid positive operator-valued measure~(POVM) and derive bounds on both the deviation of the average success probability and violations of the no-error condition.
These bounds are expressed in terms of quantities available from the estimated Gram matrix and its SDP solution together with a characterized perturbation bound, without requiring knowledge of the ideal Gram matrix.
We validate the robustness analysis numerically, characterize finite-shot estimation errors, and examine the framework through experiments on IBM quantum hardware.

This paper is organized as follows. 
In Sec.~\ref{sec:preliminaries}, we review the SDP formulation of USD and the NSS framework. 
In Sec.~\ref{sec:main01}, we present the proposed framework for the unambiguous discrimination of pure states.
In Sec.~\ref{sec:main02}, we analyze the robustness of the resulting discrimination measurement against Gram matrix estimation errors. 
In Sec.~\ref{sec:main03}, we present numerical simulations validating the robustness analysis and characterizing finite-shot estimation errors.
We also report experiments on IBM quantum hardware.
Finally, Sec.~\ref{sec:conclusion} concludes the paper.

\section{Preliminaries}
\label{sec:preliminaries}

\subsection{Unambiguous State Discrimination}
\label{sec:pre1}

Quantum state discrimination is the task of identifying which state was prepared from a given ensemble $\mathcal{S}=\{\rho_{k}\}$.
A general measurement strategy is described by a POVM with elements $\{M_{n}\}$ satisfying $M_{n} \succeq 0$ and $\sum_{n}M_{n}=\mathbb{I}$, where $\mathbb{I}$ is the identity operator.
When the states in $\mathcal{S}$ have mutually orthogonal supports, perfect discrimination can be achieved by a projective measurement onto those supports.
For nonorthogonal states, however, no measurement can perfectly discriminate the ensemble without error. 
Avoiding misclassification therefore requires allowing an inconclusive outcome to account for the ambiguous cases.
The strategy in which every conclusive outcome is required to be correct is known as USD~\cite{Ivanovic1987, Dieks1988, Peres1988, Chefles1998}.

Formally, suppose that the states $\{\rho_{k}\}_{k=1}^{N}$ are prepared with probabilities $\{\eta_{k}\}_{k=1}^{N}$, where $\sum_{k}\eta_{k}=1$. 
We assume that the support of each $\rho_{k}$ is not entirely contained in the combined support of the remaining states, so that every $\rho_{k}$ possesses a component which lies outside the support of the other states.
In USD, the POVM elements $M_{n}~(n=1, \dots, N)$ are constructed such that outcome $n$ uniquely identifies the prepared state as $\rho_{n}$.
Since $\mathrm{Tr}(\rho_{k} M_{n})$ gives the probability of obtaining outcome $n$ when $\rho_{k}$ is prepared, this unambiguous property is expressed as
\begin{equation}
\mathrm{Tr}(\rho_{k} M_{n})=0,\quad \forall k\neq n,~~k,n \in \{1, \dots, N\},
\end{equation}
ensuring that a conclusive outcome is never produced by an incorrect state.
The remaining probability is absorbed into an additional POVM element
\begin{equation}
M_{0}=\mathbb{I}-\sum_{n=1}^{N}M_{n} \succeq 0,
\end{equation}
which corresponds to an inconclusive outcome.

The optimal POVM maximizes the average conclusive success probability under these constraints, which can be cast as the following SDP:
\begin{align}
\label{eq:usd_sdp}
\max_{\{M_{n}\}}\quad&~\sum_{n=1}^{N}\eta_{n}\mathrm{Tr}(\rho_{n} M_{n}) \nonumber \\
\text{s.t.}\quad& \mathrm{Tr}(\rho_{k} M_{n})=0,~~ \forall k \neq n, ~~k,n \in \{1, \dots, N\}, \nonumber\\
&\sum_{n=0}^{N} M_{n}=\mathbb{I},\quad  M_{n}\succeq0, ~~ \forall n \in \{0, 1, \dots, N\}.
\end{align}
This SDP serves as the starting point for the NSS-based formulation discussed in Sec.~\ref{sec:main01}.

\subsection{NISQ SDP Solver}
\label{sec:pre2}

NSS is a hybrid quantum-classical algorithm for SDPs~\cite{Bharti2022}, particularly suited to problems in which the quantities defining the SDP can be efficiently estimated using quantum circuits.
Let us consider the standard SDP of the form
\begin{align}
\label{eq:stdSDP}
\min\quad&~\mathrm{Tr}(CX) \nonumber \\
\text{s.t.}\quad& \mathrm{Tr}(A_{i}X)=b_{i},~~ \forall i \in \{1, \dots, m\}, \quad X \succeq 0,
\end{align}
where $C$ and $A_{i}~(1\le i \le m)$ are Hermitian matrices, and $b_{i}$ is the $i$-th element of a vector $b \in \mathbb{R}^{m}$.
In NSS, a fixed set of quantum states $\mathcal{A}=\{\ket{\psi_{i}}\}$, called the ansatz states, is prepared by quantum circuits, and the optimization variable $X$ is restricted to the operator subspace spanned by $\{\ket{\psi_{i}}\bra{\psi_{j}}\}$ via the hybrid operator
\begin{equation}
\label{eq:hybrid_operator}
X_{\mathcal{A}}=\sum_{i,j}\gamma_{i,j}\ket{\psi_{i}}\bra{\psi_{j}},
\end{equation}
where $\gamma=(\gamma_{i,j})$ is a positive semidefinite coefficient matrix.
For each operator $O$ in Eq.~\eqref{eq:stdSDP}~(i.e., $C$ and the $A_{i}$), $\mathrm{Tr}(OX_{\mathcal{A}})$ is given by $\sum_{i,j}\gamma_{i,j}\bra{\psi_{j}}O\ket{\psi_{i}}$, where the matrix elements $\bra{\psi_{j}}O\ket{\psi_{i}}$ are estimated using quantum circuits.
Together with the condition $\gamma \succeq 0$, which ensures $X_{\mathcal{A}} \succeq 0$, the original SDP becomes an SDP over the parameter matrix $\gamma$, whose size is determined by the number of ansatz states.

NSS differs in an important respect from variational quantum algorithms~(VQAs)~\cite{Cerezo2021}, another well-known class of hybrid quantum-classical algorithms.
In VQAs, the classical optimization variable is the parameter vector $\bm{\theta}$ of a quantum circuit $U(\bm{\theta})$, and the objective is a generally nonconvex function of $\bm{\theta}$, leading to optimization landscapes with local minima that may be far from the global optimum~\cite{Bittel2021, Larocca2025}.
In NSS, by contrast, the optimization is convex, thereby avoiding the issue of local minima.
When the ansatz set $\{\ket{\psi_{i}}\}$ is well chosen, NSS can reduce the original SDP to a smaller convex problem that can be solved efficiently.
However, this efficiency critically depends on the choice of $\{\ket{\psi_{i}}\}$.
The ansatz subspace must be expressive enough to contain (or closely approximate) an optimal solution. 
Otherwise, NSS can only recover a suboptimal approximation to the solution.
The choice of ansatz states is therefore a central practical issue in any application of NSS.


\section{NSS-Based Framework for Unambiguous Discrimination of Pure States}
\label{sec:main01}

\begin{figure}[t]
\centering
\includegraphics[width=0.47\textwidth]{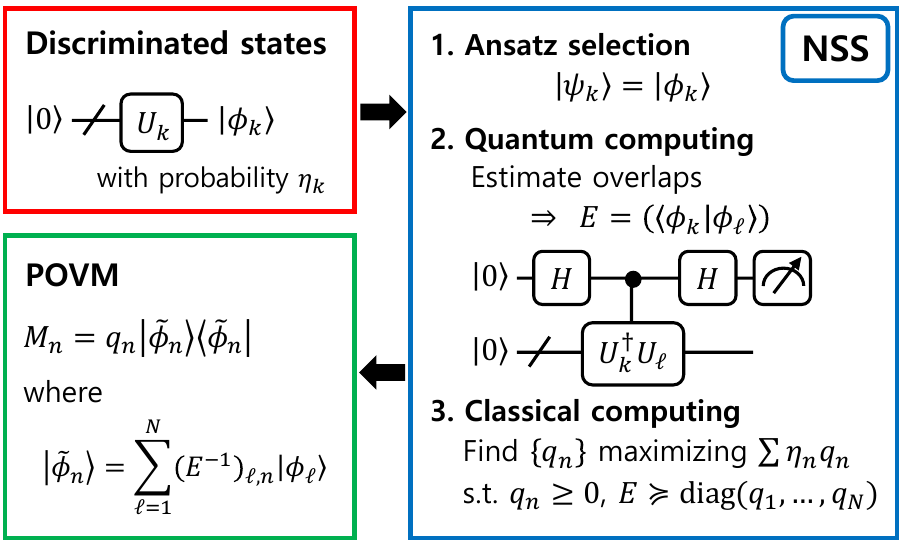}
\caption{
Hybrid quantum-classical workflow for NSS-based pure-state unambiguous discrimination. 
Discriminated states $\{\ket{\phi_{k}}\}$~(top left), prepared by $\{U_{k}\}$ with probabilities $\{\eta_{k}\}$, are processed within the NSS framework~(right) in three steps.
(1) The ansatz states are set to coincide with the discriminated states, $\ket{\psi_{k}}=\ket{\phi_{k}}$. 
(2) The Gram matrix $E$, with entries $E_{k,\ell}=\langle\phi_{k}|\phi_{\ell}\rangle$, is estimated via quantum circuits.
The illustrated Hadamard-test circuit estimates the real part of the overlap.
(3) The reduced SDP in Eq.~(\ref{eq:NSS_SDP_reduced}) is solved classically to obtain the scalars $\{q_{n}\}$.
The optimal POVM~(bottom left) is constructed from $E$ and $\{q_{n}\}$ through Eqs.~(\ref{eq:reciprocalstate}) and (\ref{eq:reciprocalPOVM}).
}
\label{Figure01}
\end{figure}

We consider the problem of unambiguously discriminating between pure states $\{\ket{\phi_{k}}\}_{k=1}^{N}$.
As shown in Sec.~\ref{sec:pre1}, the optimal POVM can be obtained by solving the SDP in Eq.~(\ref{eq:usd_sdp}).
When classical vector representations of the states $\ket{\phi_{k}}$ are available, the quantities required to formulate this SDP can be directly computed.
However, when the states are specified only through their preparation circuits, obtaining such vector representations becomes increasingly difficult as the system size grows.
This setting is common in quantum information processing, where states are routinely generated as outputs of quantum circuits rather than provided as explicit vector representations.
In such settings, a hybrid approach that estimates the required quantities using quantum circuits avoids the need to explicitly reconstruct the state vectors, providing a natural setting for applying the NSS framework.

Indeed, Bharti \textit{et al.}~\cite{Bharti2022} previously applied the NSS framework to USD by parameterizing each conclusive POVM element as a hybrid operator analogous to Eq.~(\ref{eq:hybrid_operator}),
\begin{equation}
\label{eq:hybrid_povm}
M_{n}=\sum_{i,j}\gamma_{i,j}^{(n)}\ket{\psi_{i}}\bra{\psi_{j}}, \quad n=1,\dots,N,
\end{equation}
where $\{\ket{\psi_{i}}\}$ is a fixed set of ansatz states and $\gamma^{(n)}=(\gamma_{i,j}^{(n)}) \succeq 0$.
Suppose that $\rho_{k}=\ket{\phi_{k}}\bra{\phi_{k}}$ admits the hybrid representation
\begin{equation}
\rho_{k}=\sum_{i,j}\beta_{i,j}^{(k)}\ket{\psi_{i}}\bra{\psi_{j}},
\end{equation}
with $\beta^{(k)}=(\beta_{i,j}^{(k)}) \succeq 0$.
The resulting NSS problem for USD reads
\begin{align}
\label{eq:NSS_SDP_general}
\max_{\{\gamma^{(n)}\}}\quad&\sum_{n=1}^{N}\eta_{n} \mathrm{Tr}(E\beta^{(n)}E\gamma^{(n)})\nonumber \\
\text{s.t.}\quad& {\rm Tr}(E\beta^{(k)}E\gamma^{(n)})=0~~ \forall k \neq n, ~~k,n \in \{1, \dots, N\}, \nonumber\\
&E-\sum_{n=1}^{N}E\gamma^{(n)}E\succeq0, \quad \gamma^{(n)}\succeq0, ~~ \forall n \in \{1, \dots, N\},
\end{align}
where $E=(E_{i,j})$, with $E_{i,j}=\langle\psi_{i}|\psi_{j}\rangle$, is the Gram matrix of the ansatz states, whose entries are estimated using quantum circuits.
The optimal USD measurement can be represented within the subspace spanned by the discriminated states~\cite{Chefles1998}.
Hence, if the ansatz subspace contains the span of the discriminated states, solving Eq.~(\ref{eq:NSS_SDP_general}) can recover the optimal POVM.
In this formulation, each conclusive POVM element is treated as a matrix-valued optimization variable.

While Eq.~(\ref{eq:NSS_SDP_general}) provides a general NSS-based formulation of pure-state USD, the known structure of the optimal USD measurement for linearly independent pure states allows this formulation to be simplified through a problem-specific choice of ansatz.
We therefore take the discriminated states themselves as the ansatz, that is, $\ket{\psi_{k}}=\ket{\phi_{k}}$ for $1 \le k \le N$.
Throughout this work, we assume the discriminated states to be linearly independent, equivalently, that their Gram matrix $E$ is positive definite.
This ensures that all states can be unambiguously discriminated and that $E$ is invertible.
With this choice of ansatz, the hybrid representation of $\rho_{k}=\ket{\phi_{k}}\bra{\phi_{k}}$ reduces to $\beta^{(k)}=e_{k} e_{k}^{\dagger}$, where $e_{k}$ denotes the $k$-th standard basis vector.
Linear independence further guarantees the existence of a uniquely determined dual basis within the span of the discriminated states, called the reciprocal states $\{\ket{\tilde{\phi}_{k}}\}_{k=1}^{N}$, defined by the biorthogonality relations $\langle\tilde{\phi}_{k}|\phi_{\ell}\rangle=\delta_{k,\ell}$ for $1 \le k,\ell \le N$, equivalently
\begin{align}
\label{eq:reciprocalstate}
\ket{\tilde{\phi}_{k}}=\sum_{\ell=1}^{N}(E^{-1})_{\ell,k}\ket{\phi_{\ell}}.
\end{align}
It is well established that the conclusive elements of an optimal USD POVM take the form~\cite{Eldar2003}
\begin{equation}
\label{eq:reciprocalPOVM}
M_{n}=q_{n}\ket{\tilde{\phi}_{n}}\bra{\tilde{\phi}_{n}}, \quad n=1,\dots,N,
\end{equation}
where $\{q_{n}\}_{n=1}^{N}$ are non-negative scalars to be optimized.

Substituting the reciprocal-state POVM in Eq.~(\ref{eq:reciprocalPOVM}) into the general NSS formulation in Eq.~(\ref{eq:NSS_SDP_general}) yields several simplifications.
The POVM coefficient matrices reduce to
\begin{equation}
\gamma^{(n)}=q_{n}E^{-1}e_{n} e_{n}^{\dagger}E^{-1},
\end{equation}
each of which has rank at most one, with positivity $\gamma^{(n)} \succeq 0$ ensured by $q_{n} \ge 0$.
With this form, the no-error condition ${\rm Tr}(E\beta^{(k)}E\gamma^{(n)})=0$ for $k \neq n$ holds automatically, and each term $\mathrm{Tr}(E\beta^{(n)}E\gamma^{(n)})$ in the objective function then evaluates to $q_{n}$.
The completeness constraint $E-\sum_{n=1}^{N}E\gamma^{(n)}E\succeq0$ becomes
\begin{equation}
E- {\rm diag}(q_1,\ldots,q_N) \succeq 0.
\end{equation}
The resulting SDP, which coincides with the established formulation in Ref.~\cite{Eldar2003}, is therefore
\begin{align}
\label{eq:NSS_SDP_reduced}
    \max_{\{q_{n}\}}\quad& \sum_{n=1}^{N}\eta_{n} q_{n} \nonumber\\
    \text{s.t.}\quad& q_{n}\ge 0, \quad n \in \{1, \dots, N\}, \nonumber\\
    & E- {\rm diag}(q_1,\ldots,q_N) \succeq 0.
\end{align}
Based on this formulation, we propose an NSS-based hybrid framework for pure-state unambiguous discrimination that requires only the Gram matrix of the discriminated states as input to the classical optimization.
The Gram matrix $E$ is estimated on a quantum device via overlap measurements between pairs of discriminated states, while the SDP in Eq.~(\ref{eq:NSS_SDP_reduced}) is solved classically to determine $\{q_{n}\}$. 
The discrimination POVM is then determined through Eqs.~(\ref{eq:reciprocalstate})--(\ref{eq:reciprocalPOVM}).
The resulting hybrid quantum-classical workflow is summarized in Fig.~\ref{Figure01}.

Compared with the original NSS-based USD formulation, the reduced formulation substantially simplifies the classical optimization.
With the discriminated states used as the ansatz in both formulations, the original formulation in Eq.~\eqref{eq:NSS_SDP_general} involves $N$ positive-semidefinite $N \times N$ matrix variables, together with $N+1$ positive-semidefinite constraints and $N(N-1)$ no-error equalities.
In contrast, Eq.~\eqref{eq:NSS_SDP_reduced} involves only $N$ nonnegative scalar variables and a single $N \times N$ positive-semidefinite constraint.
For solvers that enforce positive semidefiniteness through eigenvalue decompositions, processing each $N \times N$ matrix costs $O(N^{3})$.
The original formulation requires this operation for $N+1$ matrices per iteration, whereas the reduced formulation requires it for only one, reducing this part of the computational cost from $O(N^4)$ to $O(N^3)$.
Computing $E^{-1}$ to obtain the reciprocal-state coefficients requires a one-time $O(N^{3})$ calculation, the same order as processing a single matrix.
Although the total runtime depends on the solver and the number of iterations, the reduced formulation yields a smaller classical optimization problem while preserving the optimal USD success probability for an exact Gram matrix.

\section{Noise Analysis}
\label{sec:main02}

In NSS-based formulations of USD, the Gram matrix $E$ of the ansatz states is supplied to the classical solver from the quantum side of the workflow.
The accuracy with which $E$ is estimated therefore directly affects the quality of the resulting POVM and, in turn, the achievable success probability.
Even in the absence of hardware-induced circuit noise, the overlaps $E_{k,\ell}=\langle\phi_{k}|\phi_{\ell}\rangle$ are estimated using a finite number of measurement shots and are therefore subject to statistical error.
In this section, we analyze how such a perturbation in $E$ propagates through the workflow.
We first construct a valid POVM from the estimated Gram matrix under a bounded perturbation and then analyze the robustness of its discrimination performance.

Suppose that the quantum process in the NSS yields a noisy Gram matrix of the form
\begin{align}
E_{noisy} = E_{ideal} + \Delta E,
\end{align}
with $\Vert \Delta E \Vert_{2} \le \varepsilon$ for some $0 < \varepsilon \le 1$, where $\Vert\cdot\Vert_{2}$ denotes the operator norm for matrices and the Euclidean norm for vectors.
In constructing $E_{noisy}$, the diagonal entries are fixed to unity without measurement, while the upper-triangular off-diagonal entries are estimated from the quantum device.
The corresponding lower-triangular entries are then set to their complex conjugates, thereby enforcing Hermiticity.
Thus, $\Delta E$ is Hermitian with zero diagonal.

For circuit-prepared states, the exact Gram matrix $E_{ideal}$ may not be classically available, and $E_{noisy}$ is therefore the practically accessible quantity.
In such a setting, an upper bound on the individual overlap-estimation errors may be empirically characterized for a given quantum device and measurement setting and used to obtain a corresponding bound $\varepsilon$ on $\Vert \Delta E \Vert_{2}$.
For example, if each off-diagonal overlap has an estimation error bounded by $\delta$, then $\Vert \Delta E \Vert_{2} \le (N-1)\delta$ provides a simple conservative choice of $\varepsilon$.

Let $\hat{\lambda}_{min}$ be the minimum eigenvalue of $E_{noisy}$.
Throughout this section, we assume $\hat{\lambda}_{min} > \varepsilon$. 
This ensures that $E_{noisy}$ is positive definite and, by Weyl's inequality~\cite{Horn2012}, that $E_{ideal}$ is also positive definite.
Let $\{\hat{q}_{n}\}$ denote an optimizer of the SDP in Eq.~\eqref{eq:NSS_SDP_reduced} with $E=E_{noisy}$.
We first construct the positive operators
\begin{equation}
\label{eq:noisyPOVM}
\widetilde{M}_{n}=\hat{q}_{n}\ket{\tilde{\phi}_{n}^{noisy}}\bra{\tilde{\phi}_{n}^{noisy}}, \quad n=1,\dots,N,
\end{equation}
where
\begin{equation}
\ket{\tilde{\phi}_{k}^{noisy}}=\sum_{\ell=1}^{N}(E_{noisy}^{-1})_{\ell,k}\ket{\phi_{\ell}}.
\end{equation}
Because the estimated Gram matrix $E_{noisy}$ differs from the ideal Gram matrix $E_{ideal}$, these operators need not satisfy $\sum_{n=1}^{N}\widetilde{M}_{n} \preceq \mathbb{I}$.
We therefore define
\begin{align}
&\widehat{M}_{n} = \frac{\hat{\lambda}_{min}}{\hat{\lambda}_{min}+\varepsilon}\widetilde{M}_{n}, \quad n=1,\dots,N, \nonumber \\
&\widehat{M}_{0} = \mathbb{I}-\sum_{n=1}^{N}\widehat{M}_{n}.
\end{align}
This rescaling guarantees a valid POVM.
To see this, let $D=\mathrm{diag}(\hat{q}_{1},\ldots,\hat{q}_{N})$ and $\Phi=(\ket{\phi_{1}}, \dots, \ket{\phi_{N}})$.
The SDP constraint $D \preceq E_{noisy}$ gives
\begin{equation}
\sum_{n=1}^{N} \widetilde{M}_{n}=\Phi E_{noisy}^{-1}D E_{noisy}^{-1} \Phi^{\dagger} \preceq \Phi E_{noisy}^{-1} \Phi^{\dagger}.
\end{equation}
Since $\Phi^{\dagger}\Phi=E_{ideal} \preceq E_{noisy}+ \varepsilon \mathbb{I}$, we have
\begin{equation}
\left\Vert \Phi E_{noisy}^{-1} \Phi^{\dagger} \right\Vert_{2}= \left\Vert E_{noisy}^{-1/2} E_{ideal} E_{noisy}^{-1/2} \right\Vert_{2} \le 1+\frac{\varepsilon}{\hat{\lambda}_{min}}.
\end{equation}
Thus $\sum_{n=1}^{N}\widehat{M}_{n} \preceq \mathbb{I}$, so that $\widehat{M}_{0} \succeq 0$ and $\{\widehat{M}_{n}\}_{n=0}^{N}$ defines a valid POVM.
Although this rescaling guarantees a valid POVM, the no-error conditions need not hold exactly because of Gram-matrix estimation errors.
Theorem~\ref{Theorem2} provides upper bounds on the resulting misidentification probabilities.

We now quantify how the Gram-matrix perturbation affects the optimization underlying this construction.
For any positive semidefinite matrix $E$, define $Q(E)$ as the optimal value of the SDP in Eq.~(\ref{eq:NSS_SDP_reduced}),
\begin{equation}
\label{eq:def_Q}
Q(E)=\max_{q_{n}\ge 0}\left\{\sum_{n=1}^{N}\eta_{n} q_{n} \;\bigg|\; E-\mathrm{diag}(q_{1},\ldots,q_{N})\succeq 0\right\}.
\end{equation}
Then $Q(E_{ideal})$ is the optimal SDP value under ideal conditions, whereas $Q(E_{noisy})$ is the value obtained from the estimated Gram matrix.
Since $E_{ideal}$ may not be available in practice, we express the following bound in terms of quantities accessible from $E_{noisy}$ and the perturbation bound $\varepsilon$.
\begin{Lem}
\label{lemma1}
Suppose that $\Vert \Delta E \Vert_{2} \le \varepsilon$ and $\hat{\lambda}_{min}> \varepsilon$.
Then
\begin{equation}
\left| Q(E_{noisy})-Q(E_{ideal}) \right| \leq \frac{\varepsilon}{\hat{\lambda}_{min}}Q(E_{noisy}).    
\end{equation}
\end{Lem}
\begin{proof}
Let $r=\varepsilon/\hat{\lambda}_{min}$.
Then $0 < r < 1$.
Since $E_{noisy} \succeq \hat{\lambda}_{min}\mathbb{I}$ and $\Vert \Delta E \Vert_{2} \le \varepsilon$, we have
\begin{equation}
-rE_{noisy} \preceq \Delta E \preceq rE_{noisy}.
\end{equation}
Using $E_{ideal}=E_{noisy}-\Delta E$, it follows that
\begin{equation}
\label{eq:lem_positivity_ineq}
(1-r)E_{noisy} \preceq E_{ideal} \preceq (1+r)E_{noisy}.
\end{equation}
Since $Q(E)$ is monotone with respect to the order of positive-semidefinite matrices and satisfies $Q(cE)=cQ(E)$ for $c>0$, Eq.~\eqref{eq:lem_positivity_ineq} gives
\begin{equation}
(1-r)Q(E_{noisy}) \le Q(E_{ideal}) \le (1+r)Q(E_{noisy}).
\end{equation}
Therefore,
\begin{equation}
\left| Q(E_{noisy})-Q(E_{ideal}) \right| \leq rQ(E_{noisy}),    
\end{equation}
which completes the proof.
\end{proof}

We define the average success probability for the constructed POVM $\{\widehat{M}_{n}\}$ as
\begin{equation}
P(E_{noisy})=\sum_{n=1}^{N}\eta_{n}\langle \phi_{n}| \widehat{M}_{n} | \phi_{n} \rangle.
\end{equation}
For the ideal benchmark, we use the unscaled optimal USD POVM constructed from $E_{ideal}$, so that $P(E_{ideal})=Q(E_{ideal})$.
The following theorem quantifies the effect of Gram-matrix perturbations on both the average success probability and violations of the no-error condition of the resulting POVM.
\begin{Thm}
\label{Theorem2}
Suppose that $\Delta E$ is Hermitian with zero diagonal, $\Vert \Delta E \Vert_{2} \le \varepsilon$ and $\hat{\lambda}_{min}> \varepsilon$.
Let 
\begin{equation}
\mu=\frac{\hat{\lambda}_{min}}{\hat{\lambda}_{min}+\varepsilon}, \quad \nu=\frac{\varepsilon}{\hat{\lambda}_{min}}Q(E_{noisy}),
\end{equation}
and, for $n,k \in \{1, \dots, N\}$, define
\begin{equation}
s_{n,k}=\sqrt{\left\Vert E_{noisy}^{-1}e_{n} \right\Vert_{2}^{2} - \left|(E_{noisy}^{-1})_{n,k}\right|^{2}},
\end{equation}
where $e_n$ is the $n$-th standard basis vector.
Then
\begin{align}
\label{eq:deviationbound}
&\left| P(E_{noisy})-P(E_{ideal}) \right| \nonumber \\
&\leq \mu(\Gamma+\nu)+(1-\mu)\min\{1, Q(E_{noisy})+\nu\},
\end{align}
where
\begin{align}
&\Gamma = 2\varepsilon\sum_{n=1}^{N}\eta_{n}\hat{q}_{n}s_{n,n}+\varepsilon^{2}\sum_{n=1}^{N}\eta_{n}\hat{q}_{n}s_{n,n}^{2}.
\end{align}
Moreover, for $k \neq n$, the probability of incorrectly assigning outcome $n$ to state $\ket{\phi_{k}}$ satisfies
\begin{equation}
\label{eq:errorbound}
\langle \phi_{k}| \widehat{M}_{n} | \phi_{k} \rangle \le \mu\varepsilon^{2}\hat{q}_{n}s_{n,k}^{2}.
\end{equation}
\end{Thm}
\begin{proof}
Substituting $E_{ideal}=E_{noisy}-\Delta{E}$ gives
\begin{align}
\langle\tilde{\phi}_{n}^{noisy}|\phi_{k}\rangle 
&= \left(E_{noisy}^{-1} \cdot E_{ideal}\right)_{n,k} \nonumber \\
&= \delta_{n,k} - \left(E_{noisy}^{-1} \cdot \Delta E\right)_{n,k}.
\end{align}
Define $B = E_{noisy}^{-1} \Delta E$.
Then
\begin{align}
\langle \phi_{n}| \widehat{M}_{n} | \phi_{n} \rangle
&= \mu\hat{q}_{n} \left|1-B_{n,n}\right|^{2} \nonumber \\
&= \mu(\hat{q}_{n} - 2\hat{q}_{n}{\rm Re}(B_{n,n})+\hat{q}_{n}| B_{n,n}|^{2}),
\end{align}
and therefore
\begin{align}
P(E_{noisy})=& \mu Q(E_{noisy}) \nonumber \\
&- 2\mu\sum_{n=1}^{N}\eta_{n}\hat{q}_{n}{\rm Re}(B_{n,n})+\mu\sum_{n=1}^{N}\eta_{n}\hat{q}_{n}| B_{n,n}|^{2}.
\end{align}
Since $(\Delta E)_{i,i}=0$ for all $i \in \{1, 2, \dots, N\}$, we obtain
\begin{equation}
B_{n,k}=\sum_{j \neq k}(E^{-1}_{noisy})_{n,j}(\Delta E)_{j,k}.
\end{equation}
By the Cauchy-Schwarz inequality,
\begin{align}
|B_{n,k}| &\le \left(\sum_{j \neq k}|(E^{-1}_{noisy})_{n,j}|^{2}\right)^{1/2}\cdot\left(\sum_{j \neq k}|(\Delta E)_{j,k}|^{2}\right)^{1/2} \nonumber \\
&=s_{n,k} \Vert \Delta E e_{k} \Vert_{2} \nonumber \\
&\le s_{n,k}\varepsilon,
\end{align}
where the equality follows from the Hermiticity of $E_{noisy}^{-1}$ and $(\Delta E)_{k,k}=0$, and the last inequality follows from $\Vert \Delta E e_{k} \Vert_{2} \le \Vert \Delta E\Vert_{2} \Vert e_{k} \Vert_{2}$ and $\Vert \Delta E \Vert_{2} \le \varepsilon$.
Hence,
\begin{align}
\label{eq:thm_eq1}
&\left| P(E_{noisy})-\mu Q(E_{noisy})\right| \le \mu\Gamma.
\end{align}
Combining Eq.~(\ref{eq:thm_eq1}) with Lemma~\ref{lemma1} and $P(E_{ideal})=Q(E_{ideal})$ gives
\begin{align}
\label{eq:thm2_ineq1_bound}
&\left| P(E_{noisy})-P(E_{ideal})\right| \nonumber \\
&\le \left| P(E_{noisy})-\mu Q(E_{noisy})\right| + \mu\left| Q(E_{noisy})-Q(E_{ideal})\right| \nonumber \\
& \quad +(1-\mu)P(E_{ideal})\nonumber \\
&\le \mu(\Gamma+\nu)+(1-\mu)P(E_{ideal}).
\end{align}
Since Lemma~\ref{lemma1} and $P(E_{ideal})\le 1$ imply 
\begin{equation}
P(E_{ideal})\le \min\{1, Q(E_{noisy})+\nu\},
\end{equation}
Eq.~\eqref{eq:thm2_ineq1_bound} yields the inequality in Eq.~\eqref{eq:deviationbound}.
We also note that, for $k \neq n$,
\begin{equation}
\langle \phi_{k}| \widehat{M}_{n} | \phi_{k} \rangle = \mu\hat{q}_{n} \left|B_{n,k}\right|^{2} \le \mu\varepsilon^{2}\hat{q}_{n}s_{n,k}^{2}.
\end{equation}
\end{proof}

\begin{figure*}[t]
\centering
\includegraphics[width=0.9\textwidth]{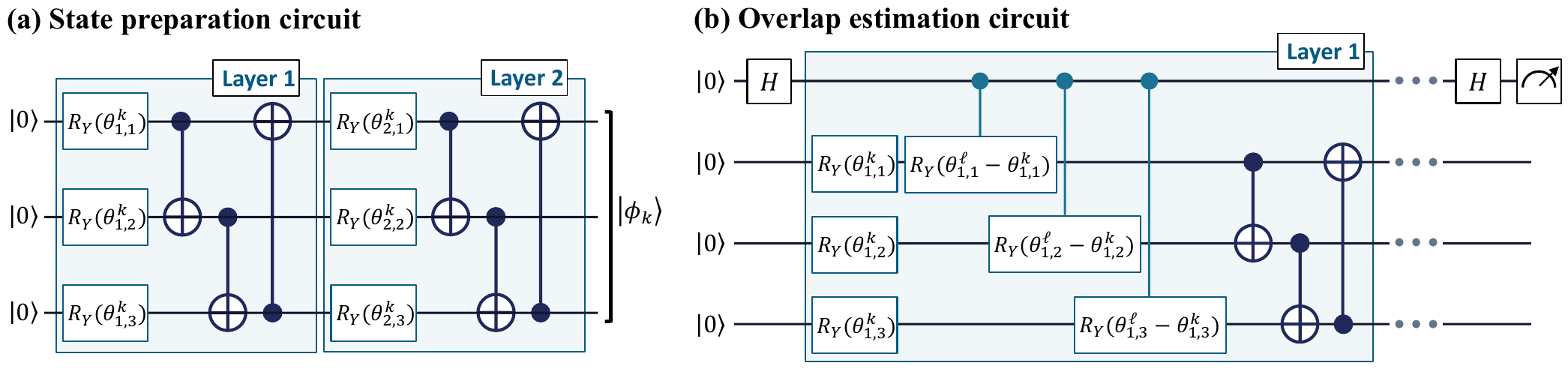}
\caption{
Quantum circuits used in the numerical simulations.
(a) State preparation circuit for generating the discriminated state $\ket{\phi_{k}}$ using a three-qubit, two-layer HEA consisting of $R_{Y}$ rotations and $CNOT$ gates.
The prepared states are real-valued, and thus their pairwise overlaps $E_{k, \ell}=\langle\phi_{k}|\phi_{\ell}\rangle$ are real.
(b) Overlap estimation circuit for computing $E_{k, \ell}$, obtained from the measurement outcome on the ancilla.
}
\label{Figure02}
\end{figure*}

The bounds in Theorem~\ref{Theorem2} can be evaluated using only the estimated Gram matrix, its SDP optimizer, and the perturbation bound $\varepsilon$, without requiring knowledge of the ideal Gram matrix.
For a fixed linearly independent state ensemble, the success-probability deviation and misidentification bounds are $O(\varepsilon)$ and $O(\varepsilon^{2})$, respectively.
However, the absolute deviation does not directly indicate the magnitude of the error relative to the ideal success probability, which varies across state ensembles.
We therefore also consider the relative error.
By Lemma~\ref{lemma1},
\begin{equation}
\label{eq:bound_of_p}
P(E_{ideal}) = Q(E_{ideal})  \ge \left(1-\frac{\varepsilon}{\hat{\lambda}_{min}}\right)Q(E_{noisy}).
\end{equation}
Combining Eq.~\eqref{eq:bound_of_p} with Eq.~\eqref{eq:thm2_ineq1_bound} yields the following bound.
\begin{Cor}
\label{Corollary3}
Define the relative error $R_{P}$ in the success probability as
\begin{equation}
R_{P}=\frac{\left| P(E_{noisy})-P(E_{ideal}) \right|}{P(E_{ideal})}.
\end{equation}
Under the conditions of Theorem~\ref{Theorem2},
\begin{equation}
\label{eq:cor_relative_bound}
R_{P} \le \frac{\mu(\Gamma+\nu)}{Q(E_{noisy})-\nu}+(1-\mu).
\end{equation}
\end{Cor}
\begin{proof}
Dividing Eq.~\eqref{eq:thm2_ineq1_bound} by $P(E_{ideal})$ gives
\begin{equation}
R_{P} \le \frac{\mu(\Gamma+\nu)}{P(E_{ideal})}+(1-\mu),
\end{equation}
and thus Eq.~\eqref{eq:bound_of_p} completes the proof.
\end{proof}

\section{Numerical Results}
\label{sec:main03}

We numerically investigate the proposed NSS-based USD framework using ensembles of $N\in \{2,3,4\}$ states generated by a three-qubit, two-layer hardware-efficient ansatz (HEA), as shown in Fig.~\ref{Figure02}(a).
For each $N$, we randomly select one ensemble by drawing the ansatz parameters uniformly from $[0,2\pi)$, and equal prior probabilities $\eta_{n}=1/N$ are assumed throughout.
In the finite-shot simulations and hardware experiments, the off-diagonal Gram matrix elements $E_{k,\ell}=\langle\phi_{k}|\phi_{\ell}\rangle$ are estimated using the circuit shown in Fig.~\ref{Figure02}(b), while the diagonal entries are fixed to $E_{k,k}=1$ and Hermiticity is enforced by setting $E_{\ell,k}=E_{k, \ell}^{*}$.
The resulting Gram matrix is passed to the reduced SDP in Eq.~\eqref{eq:NSS_SDP_reduced}, which is solved using the splitting conic solver~(SCS)~\cite{Odonoghue2016}.
Quantum-circuit simulations are performed using the PennyLane software development kit~\cite{Bergholm2018}.
We first validate the perturbation bounds derived in Theorem~\ref{Theorem2} and Corollary~\ref{Corollary3} using synthetic Gram-matrix perturbations.
We then examine finite-shot errors in Gram-matrix estimation and evaluate the proposed framework using Gram matrices estimated on IBM quantum hardware.

\begin{figure*}[t]
\centering
\includegraphics[width=0.9\textwidth]{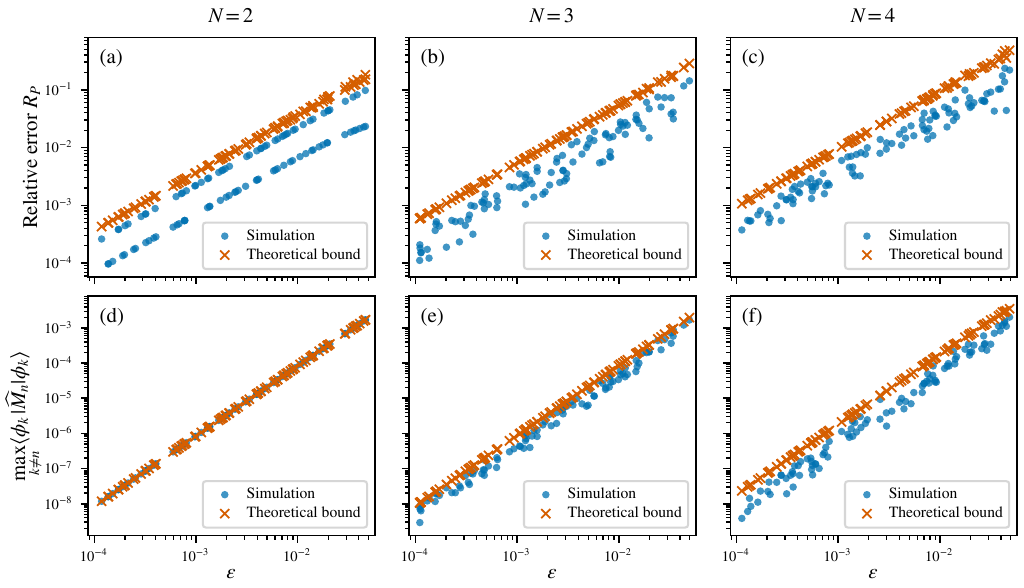}
\caption{
Numerical validation of the robustness bounds for ensembles of $N=2,3,4$ pure states, generated using the three-qubit, two-layer ansatz circuit in Fig.~\ref{Figure02}(a), with equal prior probabilities.
Panels (a)--(c) show the relative error $R_{P}$ in the average success probability and the bound in Eq.~\eqref{eq:cor_relative_bound}.
Panels (d)--(f) show the maximum misidentification probability, $\max_{k\neq n}\bra{\phi_k}\widehat{M}_n\ket{\phi_k}$, and the bound obtained by maximizing the right-hand side of Eq.~(\ref{eq:errorbound}) over $k \neq n$.
For each $N$, one randomly generated state ensemble is fixed, and 100 perturbation strengths $\varepsilon$ are sampled log-uniformly from $[10^{-4}, 5 \times 10^{-2}]$.
For each $\varepsilon$, an independent real symmetric perturbation $\Delta E$ with zero diagonal and $\Vert \Delta E \Vert_{2}=\varepsilon$ is generated.
All instances satisfy $\hat{\lambda}_{min}>\varepsilon$, ensuring that the assumptions of Theorem~\ref{Theorem2} and Corollary~\ref{Corollary3} are met.
Blue dots indicate the computed values, and orange crosses indicate the corresponding theoretical bounds.
}
\label{Figure03}
\end{figure*}

To validate the robustness analysis in Sec.~\ref{sec:main02}, we fix one randomly selected state ensemble for each $N \in \{2, 3, 4\}$ and introduce real symmetric perturbations $\Delta E$ with zero diagonal and spectral norm $\Vert \Delta E \Vert_{2}=\varepsilon$.
For each ensemble, we draw 100 perturbation strengths log-uniformly from $[10^{-4}, 5 \times 10^{-2}]$, independently generate a random perturbation for each strength, and construct $E_{noisy}=E_{ideal}+\Delta E$.
All sampled instances satisfy $\hat{\lambda}_{min}>\varepsilon$.
Figs.~\ref{Figure03}(a)--\ref{Figure03}(c) compare the relative error $R_{P}$ with the bound in Corollary~\ref{Corollary3}, while Figs.~\ref{Figure03}(d)--\ref{Figure03}(f) compare the maximum misidentification probability with the corresponding bound obtained by maximizing Eq.~\eqref{eq:errorbound} over $k \neq n$.
The computed values lie at or below their respective bounds within numerical precision.
The bounds also capture the overall dependence on the perturbation strength.
The relative errors exhibit an approximately linear trend in $\varepsilon$, whereas the maximum misidentification probabilities exhibit an approximately quadratic trend over the sampled range.

The spread of the computed errors reflects their dependence on the perturbation direction, which is not specified by the spectral norm alone.
For $N=2$, a real symmetric zero-diagonal perturbation of fixed norm has only two possible forms, differing by an overall sign.
These two perturbations produce the two branches visible in Fig.~\ref{Figure03}(a).
For $N=3,4$, the additional off-diagonal degrees of freedom allow a wider variety of perturbation directions, giving rise to a more dispersed pattern of computed errors in Figs.~\ref{Figure03}(b) and \ref{Figure03}(c).
Although the bounds incorporate information from the estimated Gram matrix, their derivation does not require knowledge of the exact perturbation and relies on inequalities that need not be saturated for a given perturbation direction.
For the tested ensembles, the bounds capture the overall dependence of the errors on the perturbation strength and provide informative estimates of their possible magnitudes.
Moreover, the misidentification bound is saturated within numerical precision for $N=2$, as shown in Fig.~\ref{Figure03}(d), and remains close to the computed values for many instances with $N=3,4$.

We next examine finite-shot errors associated with the overlap-estimation circuit in Fig.~\ref{Figure02}(b), assuming the absence of circuit noise.
For the real-valued states considered here, each ancilla measurement yields a random variable $X_{k,\ell}^{(s)} \in \{-1, +1\}$, with outcomes 0 and 1 mapped to $+1$ and $-1$, respectively, and $\mathbb{E}[X_{k,\ell}^{(s)}]=(E_{ideal})_{k, \ell}$.
Using $N_{shot}$ independent shots per overlap, we estimate each upper-triangular entry as
\begin{equation}
\label{eq:overlap_estimator}
\hat{E}_{k, \ell}=\frac{1}{N_{shot}}\sum_{s=1}^{N_{shot}}X_{k,\ell}^{(s)}, \quad k < \ell.
\end{equation}
This estimator is unbiased and has variance
\begin{equation}
\operatorname{Var}(\hat{E}_{k, \ell})=\frac{1-[(E_{ideal})_{k, \ell}]^{2}}{N_{shot}},
\end{equation}
so the standard deviation of each estimated overlap scales as $N_{shot}^{-1/2}$.

We set $(E_{noisy})_{k, \ell}=\hat{E}_{k, \ell}$ for $k< \ell$, fix the diagonal entries to unity, and obtain the lower-triangular entries by Hermitian symmetry, ensuring that $\Delta E=E_{noisy}-E_{ideal}$ is Hermitian with zero diagonal.
Since each measurement outcome $X_{k, \ell}^{(s)}$ lies in $[-1, 1]$, applying Hoeffding's inequality to the estimator in Eq.~\eqref{eq:overlap_estimator} gives, for any $\delta>0$,
\begin{equation}
\operatorname{Pr}(|(\Delta E)_{k, \ell}| \ge \delta) \le 2 \operatorname{exp} \left(-\frac{N_{shot} \delta^{2}}{2}\right), \quad k<\ell.
\end{equation}
A union bound over the $N(N-1)/2$ estimated overlaps therefore yields
\begin{equation}
\operatorname{Pr}\left(\max_{k<\ell}|(\Delta E)_{k, \ell}| \ge \delta \right) \le N(N-1)\operatorname{exp} \left(-\frac{N_{shot} \delta^{2}}{2}\right).
\end{equation}
Because $\Delta E$ is Hermitian with zero diagonal, its spectral norm is bounded by its maximum absolute row sum, so that
\begin{equation}
\Vert \Delta E \Vert_{2} \le \max_{k}\sum_{\ell \neq k}|(\Delta E)_{k, \ell}| \le (N-1)\max_{k < \ell}|(\Delta E)_{k, \ell}|.
\end{equation}
Consequently, setting $\delta = \varepsilon/(N-1)$ gives
\begin{equation}
\operatorname{Pr}(\Vert \Delta E \Vert_{2} \ge \varepsilon ) \le N(N-1)\operatorname{exp}\left(-\frac{N_{shot}\varepsilon^{2}}{2(N-1)^{2}}\right).
\end{equation}
For a prescribed failure probability $\alpha \in (0, 1)$, choosing
\begin{equation}
\varepsilon_{H}=(N-1)\sqrt{\frac{2}{N_{shot}}\ln\frac{N(N-1)}{\alpha}}
\end{equation}
ensures that
\begin{equation}
\operatorname{Pr}(\Vert \Delta E \Vert_{2} \le \varepsilon_{H}) \ge 1 - \alpha.
\end{equation}


This bound relates the number of measurement shots to the perturbation magnitude used in Sec.~\ref{sec:main02} without requiring knowledge of $E_{ideal}$.
In particular, the target accuracy $\Vert \Delta E \Vert_{2} \le \varepsilon$ is guaranteed with probability at least $1-\alpha$ if
\begin{equation}
N_{shot} \ge \frac{2(N-1)^{2}}{\varepsilon^{2}}\ln\frac{N(N-1)}{\alpha}.
\end{equation}
This is a conservative sufficient condition, rather than a minimum shot requirement for applying the framework.
For general complex-valued overlaps, the real and imaginary parts can be separately estimated using two ancilla measurement settings.
Applying Hoeffding's inequality and a union bound to both components shows that $\Vert \Delta E \Vert_{2} \le \varepsilon$ holds with probability at least $1-\alpha$ if the number of shots per component satisfies
\begin{equation}
N_{shot} \ge \frac{4(N-1)^{2}}{\varepsilon^{2}}\ln\frac{2N(N-1)}{\alpha}.
\end{equation}
Here, $N_{shot}$ denotes the number of shots per component, so each complex overlap requires a total of $2N_{shot}$ shots.

\begin{figure}[t]
\centering
\includegraphics[width=0.47\textwidth]{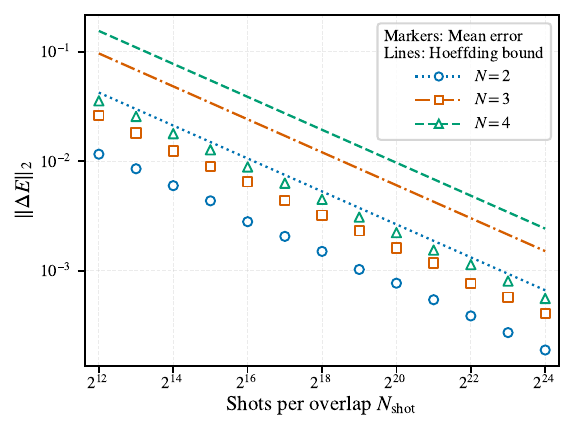}
\caption{
Finite-shot errors in Gram-matrix estimation for the same three-qubit ensembles used in Fig.~\ref{Figure03}, with $N \in \{2,3,4\}$, in the absence of circuit noise.
Markers show the mean spectral-norm error $\Vert\Delta E\Vert_{2}$ over 500 independent repetitions at each shot count, and lines show the Hoeffding bounds with failure probability $\alpha=0.05$.
}
\label{Figure04}
\end{figure}

Fig.~\ref{Figure04} illustrates the finite-shot estimation errors for the same state ensembles used in Fig.~\ref{Figure03}.
For each ensemble, we vary $N_{shot}$ from $2^{12}$ to $2^{24}$ and perform 500 independent repetitions at each shot count.
Measurement counts are sampled from binomial distributions using the exact ancilla outcome probabilities of the overlap-estimation circuit, which is statistically equivalent to independent finite-shot measurements in the absence of circuit noise.
The mean spectral-norm errors exhibit an approximately $N_{shot}^{-1/2}$ scaling, consistent with the shot-count dependence of the Hoeffding bounds.
The bounds lie above the observed mean errors, as expected from their conservative construction.
However, they provide high-probability guarantees for individual Gram-matrix estimates rather than estimates of the mean error.
They therefore serve as a sufficient connection between the measurement budget and the perturbation parameter used in Sec.~\ref{sec:main02}.
The framework can also incorporate tighter bounds on the Gram-matrix estimation error when available.

\begin{table}[t]
\centering

\begingroup
\footnotesize
\setlength{\tabcolsep}{4pt}
\begin{tabular}{c|c|c|cc|cc}
\hline\hline
& &  &\multicolumn{2}{c|}{Relative error $R_{P}$}
  & \multicolumn{2}{c}{$\max_{k \neq n} \langle \phi_{k}| \widehat{M}_{n} | \phi_{k} \rangle$} \\[-9pt]
$N$ & $\|\Delta E\|_2$ & $\hat{\lambda}_{min}$
& Value & Bound
& Value & Bound \\
\hline
2 & 0.0229 & 0.7035 & 0.0138 & 0.0790 & $4.28 \times 10^{-4}$ & $4.28 \times 10^{-4}$ \\
3 & 0.0307 & 0.5055 & 0.0213 & 0.1514 & $5.33 \times 10^{-4}$ & $7.62 \times 10^{-4}$ \\
4 & 0.0865 & 0.3465 & 0.0682 & 0.6001 & $3.728 \times 10^{-3}$ & $9.805 \times 10^{-3}$ \\
\hline\hline
\end{tabular}
\endgroup
\caption{
Performance of the POVMs constructed using Gram matrices estimated on \texttt{ibm\_pittsburgh} with $2^{14}$ shots per overlap for the ensembles used in Fig.~\ref{Figure03}.
The error $\|\Delta E\|_2$ is evaluated against the ideal Gram matrix and used for POVM rescaling and bound evaluation.
The performance is evaluated classically for the reference states and compared with the corresponding theoretical bounds.
Values are rounded for presentation.
}
\label{table01}
\end{table}

Finally, we demonstrate the proposed framework using overlap measurements performed on the IBM quantum processor \texttt{ibm\_pittsburgh}.
For the state ensembles considered in Fig.~\ref{Figure03}, we allocate $2^{14}$ shots to each upper-triangular off-diagonal overlap.
The resulting hardware-estimated Gram matrices are passed to the reduced SDP, and the POVMs are constructed with the rescaling described in Sec.~\ref{sec:main02}.
For this benchmark, the perturbation magnitude is evaluated against the known ideal Gram matrix and used in the rescaling and theoretical bounds.
Table~\ref{table01} summarizes the performance of the constructed POVMs, evaluated classically for the reference states.
All three instances satisfy $\hat{\lambda}_{min} > \varepsilon=\Vert \Delta E \Vert_{2}$, and the constructed POVMs satisfy positivity and completeness within numerical tolerance.
The relative errors in the average success probability range from approximately $1.38\%$ to $6.82\%$, while the maximum misidentification probabilities remain below $3.73 \times 10^{-3}$.
Both quantities remain at or below their corresponding theoretical bounds.
These results support the potential applicability of the proposed framework to noisy quantum hardware by demonstrating that hardware-estimated Gram matrices can be used to construct valid POVMs with quantitatively assessed performance.
In practical applications where the ideal Gram matrix is unknown, however, a reliable upper bound on the estimation error, accounting for both finite-shot fluctuations and device imperfections, must be obtained independently.
Establishing such bounds and implementing the constructed POVMs on hardware remain directions for future work.

\section{Conclusion}
\label{sec:conclusion}

In this work, we investigated the NSS-based approach to finding optimal measurements for USD of linearly independent pure states.
The NSS framework is naturally suited to this task when the states are not available as explicit classical descriptions but are instead accessed through their pairwise overlaps, as is the case for circuit-prepared states.
We showed that taking the discriminated states themselves as the ansatz, together with the reciprocal-state structure of the optimal POVM, reduces the original matrix-variable formulation to an optimization over $N$ scalar variables, thereby simplifying the classical optimization while preserving the optimal USD success probability for an exact Gram matrix.
We further analyzed the robustness of the resulting measurement against errors in the estimated Gram matrix.
In particular, we derived bounds on both the deviation of the average success probability and the violation of the no-error condition that can be evaluated from the estimated Gram matrix, its SDP solution, and a characterized upper bound on the Gram-matrix error, without requiring knowledge of the ideal Gram matrix.
These results establish the robustness of the proposed framework under bounded estimation error and, importantly, provide a practical means of assessing this robustness from quantities available in an experimental implementation.
Numerical simulations validated the derived bounds and characterized the finite-shot contribution to the Gram-matrix estimation error. 
We also estimated pairwise overlaps through experiments on IBM quantum hardware, and then assessed the performance of the resulting POVMs constructed from these estimates.

An important direction for future work is to improve the practicality of the proposed framework.
First, constructing the Gram matrix requires estimating all pairwise overlaps between the discriminated states, resulting in a number of overlap evaluations that grows quadratically with the number of states.
This measurement cost may become a significant bottleneck for large state ensembles.
Developing more efficient strategies for estimating the Gram matrix would therefore directly improve the scalability of the framework.
It would be interesting to investigate whether techniques such as classical shadows~\cite{Huang2020} can be adapted to reduce the measurement overhead associated with estimating the required overlaps.
Second, even after the discrimination POVM has been determined from the estimated Gram matrix, efficiently realizing the resulting measurement on quantum hardware remains a separate challenge.
General POVMs can be implemented by introducing ancillary systems and performing projective measurements~\cite{Nielsen2010}, but such implementations typically require additional qubits and quantum gates, increasing circuit depth and resource overhead and potentially limiting their practicality on near-term quantum devices.
Developing hardware-efficient realizations of the resulting POVMs is therefore another important step toward practical implementation of the proposed framework.

Another direction for future work is to extend the proposed framework to settings in which additional information about state-preparation noise is available.
In the present framework, the discrimination task is defined with respect to the intended pure-state ensemble, while imperfections affecting the overlap estimation are represented through the perturbed Gram matrix.
This provides a natural reference when the underlying noise process is not known a priori or cannot be characterized reliably, since in such cases the actual noisy-state ensemble itself cannot be specified with sufficient confidence for direct optimization of the measurement.
When the state-preparation noise can instead be reliably modeled or experimentally characterized, the resulting states may be treated as a noisy, generally mixed-state ensemble, and this additional information could be incorporated into the optimization of the discrimination measurement.
Developing such a noise-aware extension may enable the construction of POVMs better adapted to the actual noisy ensemble.
Beyond the discrimination of individual quantum states, another promising direction is to extend the framework to quantum classification tasks.
In such settings, the objective is to assign an input state to one of several classes rather than to identify a particular state.
This can be naturally related to set discrimination~\cite{zhang2002}, where each class is represented by a set of quantum states and the measurement is designed to distinguish between these sets.
Extending the present NSS-based formulation to this setting may therefore enable a measurement-based approach to quantum classification, particularly for circuit-prepared states whose relevant overlaps can be estimated on a quantum device.

\section*{Author Contributions}
All authors jointly conceived the project. 
J.H. and M.C. developed the theoretical results and proofs and performed the numerical simulations. 
J.H. conducted the IBM quantum hardware experiments and wrote the initial draft.
M.C. further developed and revised the manuscript.
All authors discussed the results and contributed to reviewing and revising the manuscript.
Large language models were used to assist in discussing and checking proofs, drafting and editing the manuscript, and reviewing and modifying the simulation code.
All authors have reviewed and verified the resulting text, proofs, and code and take full responsibility for the content of this work.

\section*{Acknowledgements}
This research was supported by the Korea Institute of Science and Technology Information~(KISTI) R\&D (K26L3M3C3, K26L5M1C2) and the National Research Council of Science \& Technology~(NST) (GTL25011-000).
J. H. was also supported by Creation of the Quantum Information Science R\&D Ecosystem (2022M3H3A106307411) through the National Research Foundation of Korea (NRF) funded by the Korean government (Ministry of Science and ICT).
Additional support was provided by the Creation of the Quantum Information Science R\&D Ecosystem (Based on Human Resources) through the NRF funded by the Korean government (RS-2023-NR057243).


\bibliography{refs}

\end{document}